\pdfoutput=1
\documentclass[letterpaper,abstract=true,DIV=14,titlepage=false]{scrartcl}
\usepackage[T1]{fontenc}
\usepackage[utf8]{inputenc}
\usepackage[tt=false]{libertine}
\usepackage[hyphens]{url}
\usepackage{graphicx}
\usepackage{xcolor}
\usepackage{caption}
\usepackage{float,flafter,placeins}
\usepackage{booktabs}
\usepackage{amsmath,amssymb,amsthm}
\usepackage{tikz}
\usetikzlibrary{arrows.meta,positioning}
\usepackage{microtype}
\usepackage{floatsettings}

\usepackage[datamodel=acmdatamodel,style=acmauthoryear,backend=biber,
  giveninits=true,uniquename=init,uniquelist=false]{biblatex}

\usepackage[colorlinks=true,citecolor=blue,linkcolor=blue,urlcolor=blue]{hyperref}
\hypersetup{
  pdftitle={Fast Iterative Combinatorial Auctions via Deep Reinforcement Learning},
  pdfauthor={Luis F. Paul Tosta, Quinn Arnold, Gianluca Brero}
}

\newcommand{\ind}{\mathbf{1}}

\newtheorem{theorem}{Theorem}
\newtheorem{proposition}[theorem]{Proposition}

\newcommand{\authorblock}[3]{%
  \begin{minipage}[t]{0.315\textwidth}\centering\normalfont\small
    \setlength{\parskip}{0pt}%
    {\large #1\par}
    \if\relax\detokenize{#2}\relax\else
      \vspace{2pt}%
      {\small #2\par}
    \fi
    \if\relax\detokenize{#3}\relax\else
      \vspace{1pt}%
      {\footnotesize\hypersetup{urlcolor=black}\urlstyle{same}%
        \href{mailto:#3}{\nolinkurl{#3}}\par}
    \fi
    \vspace{10pt}%
  \end{minipage}}

\title{Fast Iterative Combinatorial Auctions via Deep Reinforcement Learning}
\author{%
  \begin{tabular}{@{}c@{\hspace{0.015\textwidth}}c@{\hspace{0.015\textwidth}}c@{}}
  \authorblock{Luis F. Paul Tosta}{Bryant University}{lpaul3@bryant.edu} &
  \authorblock{Quinn Arnold}{Bryant University}{qarnold@bryant.edu} &
  \authorblock{Gianluca Brero}{Bryant University}{gbrero@bryant.edu}
  \end{tabular}}
\date{October 8, 2026}

\begin{document}
\maketitle
\vspace{-20pt}

\begin{abstract}
Iterative combinatorial auctions are used in spectrum allocation, transportation, and industrial procurement, but rounds are costly.
We introduce \emph{FastCombo}, a deep reinforcement learning approach to sequential price discovery in combinatorial auctions.
Bayesian pricing has sharply reduced auction rounds even against subgradient benchmarks tuned with hindsight for each instance, but relies on a one-step objective and repeated online inference and optimization.
We prove that greedily maximizing immediate clearing probability can require arbitrarily more rounds in expectation than an informative policy, even with the true valuation distribution known.
FastCombo learns price policies offline, assigning credit to queries for the future decisions they enable and producing each online price update with a single neural-network forward pass.
We formulate the objective of minimizing expected rounds to clearing in single-minded auctions as a partially observable Markov decision process and train a recurrent bidder--item graph policy over complete auction trajectories.
In a one-item experiment, all five trained policies learn an informative opening, reducing expected rounds from 1.85 to 1.65 relative to the tested Bayesian one-step rule.
Across four domains from the Combinatorial Auction Test Suite, FastCombo achieves competitive clearing while reducing measured online price-selection time per auction by a factor of roughly 300--900 relative to the standard Bayesian baseline.
These results demonstrate the value of learning sequential price policies for both information acquisition and online computational efficiency.
\end{abstract}

\section{Introduction}
\label{sec:introduction}

Iterative combinatorial auctions are used in spectrum allocation, freight transportation, and industrial procurement~\citep{cramton2006combinatorial}.
In these auctions, the auctioneer repeatedly posts prices and observes demand to discover prices supporting an efficient allocation without eliciting complete valuation functions~\citep{ausubel2017practical}.
However, each round requires bidders to assess alternatives and make consequential bidding decisions, imposing deliberation and organizational costs~\citep{kwasnica2005new,parkes2006iterative}.
These costs can accumulate across many rounds: Canada's 2014 700\,MHz auction required 108 rounds over 22 business days~\citep{canada2014_700mhz}.
Aggressive price increments can shorten this process but risk overshooting clearing prices and sacrificing allocative efficiency~\citep{brero2019fast}.

Bayesian auctions demonstrate how learning from demand can substantially reduce the number of rounds~\citep{brero2018bayesian,brero2019fast}.
They maintain beliefs over bidder values and use those beliefs to compute candidate clearing prices.
Their performance is especially striking given the strength of the subgradient benchmarks: the instance-level oracle in \citet{brero2019fast} runs 100 step sizes and retrospectively selects the fastest clearing run for each auction.
Even with this advantage, that oracle required 21.3 rounds on average on the single-minded Arbitrary instances cleared by all compared methods, versus 4.5 for the general Bayesian rule.
Bayesian pricing thus provides a demanding benchmark for reducing bidder interaction.

Two challenges remain.
First, the Bayesian one-step pricing objective does not explicitly optimize the future value of information: a query can be useful because its response distinguishes valuations requiring different subsequent prices, even when it clears less often immediately.
Second, price updates require online inference and optimization; the general Bayesian rule uses repeated sampling and optimization within Monte Carlo expectation maximization (MCEM)~\citep{brero2019fast}.
These challenges motivate policies that account for future demand responses while keeping online price selection inexpensive.

We take a first step in single-minded auctions, where each bidder values one desired bundle.
Our objective is to minimize expected rounds to clearing over auctions drawn from a valuation distribution.
We prove that maximizing immediate clearing probability can be arbitrarily suboptimal for this objective, even when policies know the true valuation distribution.
In our construction, failed greedy queries eliminate one possible valuation profile at a time, while informative queries divide the remaining profiles in half.
We formulate this sequential objective as a partially observable Markov decision process (POMDP) and introduce \emph{FastCombo}, a recurrent graph policy trained offline with deep reinforcement learning.
Training over complete auctions allows a price to receive credit for the information it reveals and the later clearing it enables.
At deployment, a neural-network forward pass maps public feedback to the next item prices, replacing repeated online price optimization.
The contribution concerns price discovery; winner determination and final allocation requirements remain separate.

A controlled one-item experiment shows how FastCombo learns to use demand responses to choose its next query.
All five trained policies clear in 1.65 expected rounds, compared with 1.85 for the Bayesian one-step rule at the specified precision.
Their informative opening accepts a lower chance of immediate clearing to make the next query decisive when clearing fails.
A ten-seed benchmark on the Combinatorial Auction Test Suite (CATS)~\citep{leyton2000towards} then evaluates the clearing--computation tradeoff: FastCombo attains competitive clearing, while the standard Bayesian baseline uses roughly 300--900 times as much measured online price-selection computation per auction, excluding offline training.
Bayesian pricing uses fewer queries on average in three of four domains among instances cleared by all compared methods.
Both experiments use the same FastCombo architecture, with separately trained policies.
The evidence establishes informative query selection in the one-item example and a computational benefit on CATS; broader query savings and scalability to large spectrum auctions remain open.

\section{Preliminaries}
\label{sec:preliminaries}

We consider \(m\) indivisible items, indexed by \([m]=\{1,\ldots,m\}\), and \(n\) single-minded bidders, indexed by \([n]=\{1,\ldots,n\}\).
Bidder \(i\) has a desired bundle \(\varnothing\ne S_i\subseteq[m]\) and values any bundle containing \(S_i\) at \(v_i\geq0\), with zero value otherwise.
Write \(\mathcal S=(S_1,\ldots,S_n)\) for the bundle profile and \(v=(v_1,\ldots,v_n)\) for the value profile.

The auctioneer posts a common vector of nonnegative item prices \(p\in\mathbb R_+^m\), where \(p_j\) is the price of item \(j\).
Prices add across items, so bidder \(i\)'s desired bundle costs
\begin{equation}
p(S_i)=\sum_{j\in S_i}p_j.
\label{eq:bundle-cost}
\end{equation}
Under myopic best-response bidding, its demand report is
\begin{equation}
q_i(p)=
\begin{cases}
S_i,&p(S_i)\leq v_i,\\
\varnothing,&p(S_i)>v_i.
\end{cases}
\label{eq:demand-report}
\end{equation}
Thus a bidder reports its desired bundle when its value covers the quoted price, with ties favoring demand unless stated otherwise.
A zero-price demand query would reveal the bundle profile because \(q_i(0)=S_i\) for every bidder.
Our evaluated protocol supplies \(\mathcal S\) directly before the counted rounds; values remain private.

An efficient allocation maximizes total bidder value while assigning each item at most once.
Letting \(x_i\) indicate whether bidder \(i\) receives \(S_i\), efficient welfare is
\begin{equation}
W^\star(\mathcal S,v)
=\begin{aligned}[t]
\max_{x\in\{0,1\}^n}\quad &\sum_{i=1}^{n}v_i x_i\\
\text{subject to}\quad &\sum_{i:j\in S_i}x_i\leq1
\quad\forall j\in[m].
\end{aligned}
\label{eq:wdp}
\end{equation}
To assess whether prices support an efficient allocation, we use the clearing potential of \citet{brero2019fast}.
At prices \(p\), bidder \(i\)'s maximum utility is \(\max\{0,v_i-p(S_i)\}\): it can buy its bundle or abstain.
The clearing potential adds these utilities to the sum of posted item prices,
\begin{equation}
L(p;\mathcal S,v)
=\sum_{i=1}^{n}\max\{0,v_i-p(S_i)\}+\sum_{j=1}^{m}p_j.
\label{eq:clearing-objective}
\end{equation}
This quantity upper-bounds the welfare of every feasible allocation, so
\(G(p;\mathcal S,v):=L(p;\mathcal S,v)-W^\star(\mathcal S,v)\geq0\).
At zero gap, an efficient allocation gives every bidder a utility-maximizing choice and leaves only zero-price items unsold; an indifferent bidder can receive either optimal choice.
This is the welfare and linear-programming duality characterization of clearing prices~\citep{bikhchandani2002package,brero2018bayesian}.
Linear clearing prices need not exist for every bundle profile and valuation vector.
Appendix~\ref{app:clearing} gives the gap decomposition.
For numerical evaluation, we count a posted vector as clearing when
\begin{equation}
G(p;\mathcal S,v)<\varepsilon_{\mathrm{clr}},
\label{eq:clearing}
\end{equation}
for a fixed tolerance \(\varepsilon_{\mathrm{clr}}>0\).

At each round, a price policy \(\pi\) chooses an item-price vector using the public bundle profile and preceding price--demand observations.
Feedback consists of the new demand reports and whether the prices clear.
Each posted vector constitutes one demand query and one auction round; clearing ends the interaction.
Our objective is to reach clearing in as few rounds as possible.

\section{The Value of Informative Queries}
\label{sec:informative-theory}

Fix the public bundle profile \(\mathcal S\) and draw values \(v\) once from a distribution \(F\), keeping them fixed throughout the auction.
In the analytical comparisons below, both policies know the true distribution \(F\) and observe demand and clearing feedback; the realized values remain hidden.
For exact clearing, \(G=0\), let \(\tau_\pi\) count rounds up to and including the clearing round under policy \(\pi\), with \(\tau_\pi=\infty\) if clearing never occurs.
The expected number of rounds to clearing is
\begin{equation}
J_F(\pi)=\mathbb E_{v\sim F,\pi}[\tau_\pi].
\label{eq:expected-rounds}
\end{equation}
The expectation averages over values drawn from \(F\) and any policy randomization.
A \emph{greedy clearing policy} maximizes the next query's posterior clearing probability after each public history.
The following result shows that this criterion can be arbitrarily suboptimal for expected rounds to clearing, even when \(F\) is known exactly.

\begin{proposition}[The cost of greedy clearing]
\label{thm:greedy-separation}
For every integer \(K\geq2\), there is a one-item, two-bidder auction with values and prices in \([0,1]\) and a distribution \(F\) supported on \(K\) value profiles such that, when \(v\sim F\),
\(J_F(\pi_{\mathrm{greedy}})=(K+2)/3\) for every greedy clearing policy, while some deterministic informative policy satisfies
\(J_F(\pi_{\mathrm{info}})\leq\lceil\log_2(K+1)\rceil\).
Both policies know \(F\), observe demand and exact-clearing feedback, and count one round for every posted price.
Thus the ratio of greedy to optimal expected round counts is unbounded over this family.
\end{proposition}

The construction assigns disjoint clearing-price intervals to increasingly likely value profiles.
Greedy tests the most likely remaining profile; on failure, every remaining profile produces the same demand, so only one possibility is eliminated.
An informative query instead tests a middle profile, and a failed response identifies which half contains the true values.
This yields linear versus logarithmic growth in rounds; Appendix~\ref{app:greedy-separation} gives the construction and proof.

Both expected round counts average over values genuinely drawn from the same stated distribution.
The result establishes that greedy can be suboptimal; it does not assert this for every distribution or quantify the loss on typical auctions.
This motivates optimizing whole query sequences, which reinforcement learning can approximate from sampled auctions without a guarantee of finding the optimal policy.

\section{The Auction as a POMDP}
\label{sec:pomdp}

We formulate the objective of minimizing expected rounds to clearing as a partially observable Markov decision process (POMDP), enabling reinforcement learning to assign credit to queries for the later decisions they support.
For each auction domain \(d\), an instance \(\xi=(\mathcal S,v)\) is drawn from a distribution \(\nu_d\) and remains fixed throughout the auction.
Conditional on \(\mathcal S\), this induces a valuation distribution of the form \(F\) above.
The bundle profile is public and the values are hidden.
The nonterminal state comprises this instance and the public round counter; clearing leads to an absorbing terminal state.

After \(t\) rounds, the next action posts \(p_{t+1}\in[0,p_d^{\max}]^m\), where \(p_d^{\max}>0\) is the domain's public price cap.
Let \(\ind\{A\}\) equal one if event \(A\) holds and zero otherwise.
With bundles known, binary demand is \(y_{i,t+1}=\ind\{q_i(p_{t+1})=S_i\}\); write \(y_t=(y_{1,t},\ldots,y_{n,t})\) for the demand vector.
The response also includes the clearing indicator
\(\chi_{t+1}=\ind\{G(p_{t+1};\mathcal S,v)<\varepsilon_{\mathrm{clr}}\}\).
The public history is therefore
\begin{equation}
\mathcal H_t=
\bigl(\mathcal S,(p_r,y_r,\chi_r)_{r=1}^{t}\bigr),
\qquad \mathcal H_0=(\mathcal S).
\label{eq:public-history}
\end{equation}
Observations and termination are deterministic conditional on the instance and the posted price.
A policy \(\pi(\cdot\mid\mathcal H_t)\) chooses the next price distribution from this observable history.

The natural cost of an auction round is one query.
Let \(\tau_\pi=\inf\{t\geq1:\chi_t=1\}\), with \(\tau_\pi=\infty\) if clearing never occurs.
For a round budget \(H\), define \(Q_H(\pi)=\min\{\tau_\pi,H\}\), so an unsuccessful auction is charged all \(H\) queries.
An undiscounted reward of \(-1\) for each posted price, including the clearing query, gives
\begin{equation}
\mathbb E_{\xi\sim\nu_d,\pi}\!\left[\sum_{t=0}^{Q_H(\pi)-1}r_t\right]
=-\mathbb E_{\xi\sim\nu_d,\pi}[Q_H(\pi)],
\qquad r_t=-1.
\label{eq:round-reward}
\end{equation}
Thus maximizing this return minimizes expected charged rounds.
We also report clearing probability, since a clear at the budget and a failure incur the same capped query count.
The experiments specify when training uses this unit query cost or a dense surrogate and how it handles the round limit.
Private simulator values can supply training feedback, but neither these values nor realized rewards enter the policy's public observations.

\section{FastCombo Policy}
\label{sec:policy}

FastCombo approximates a policy over public histories with recurrent item states and a graph connecting bidders to their desired items.
The recurrence summarizes successive price--demand responses; the graph allows prices to depend on conflicts between overlapping bundles.
We train the policy offline against simulated auctions and deploy it with one neural-network evaluation per query.
Appendix~\ref{app:policy-details} specifies the implementation.

\paragraph{Public features and memory.}
For item \(j\), aggregate demand is
\(D_{t,j}=\sum_{i:j\in S_i}y_{i,t}\).
Demand also yields bounds on each bidder's value: accepting a bundle tightens its lower bound using the quoted price, and rejecting it tightens its upper bound.
The bounds start at \([0,C_d]\), where \(C_d>0\) is a public valuation ceiling with \(v_i\leq C_d\).
We normalize value bounds by \(C_d\), prices by \(p_d^{\max}\), and item demand by \(n\).
A shared gated recurrent unit (GRU) maintains 32-dimensional item memory \(h_{t,j}\) from successive responses~\citep{cho2014gru}:
\begin{equation}
h_{t,j}=\operatorname{GRU}((\bar p_{t,j},\bar D_{t,j}),h_{t-1,j}),
\quad t\geq1,\qquad h_{0,j}=0.
\label{eq:recurrence-main}
\end{equation}
Bars denote normalized quantities.
The reset input precedes any response, so it does not update the recurrent state.
Neither private values nor realized training rewards enter the actor's inputs or memory.

\paragraph{Bundle graph and prices.}
Shared encoders form a 128-dimensional token for each bidder from its normalized value bounds and bundle size, and for each item from its recurrent state, the normalized round \(t/H\), and the fraction of bidders whose bundles contain it.
The graph has an edge \((i,j)\) exactly when \(j\in S_i\)~\citep{gilmer2017mpnn}.
Each of two graph passes first updates an item from its own token and the mean of its neighboring bidder tokens, then updates a bidder from its token and the mean of its newly updated neighboring item tokens.
The passes have distinct weights, shared across nodes within a pass.
The second pass conveys conflicts through overlapping bundles; these internal computations add no auction rounds.
Mean aggregation and shared bidder maps make the actor invariant to bidder order.

A shared decoder maps each final item token to a scalar \(\mu_{t,j}\).
During training, independent Normal draws with these locations and learned scales are clipped to \([-1,1]\); proximal policy optimization (PPO) scores the resulting censored distribution, including its boundary probabilities.
Evaluation is deterministic and posts
\begin{equation}
p_{t+1,j}=\frac{p_d^{\max}}{2}
\left(1+\operatorname{clip}_{[-1,1]}(\mu_{t,j})\right).
\label{eq:deterministic-price-main}
\end{equation}
Here \(\operatorname{clip}_{[a,b]}(z)=\min\{b,\max\{a,z\}\}\) truncates \(z\) to \([a,b]\).
Each call emits one simultaneous price vector, evaluated directly without price repair or a second stage.

\section{Related Work}
\label{sec:related}

Iterative combinatorial auctions use demand feedback to refine prices.
The combinatorial clock auction follows excess demand, while iBundle uses bundle-price updates~\citep{ausubel2017practical,parkes1999ibundle}.
\citet{brero2018bayesian} introduce a Bayesian clearing mechanism that maintains posterior beliefs over values and computes Bayesian price updates; \citet{brero2019fast} extend this approach to general valuation distributions with Monte Carlo expectation maximization.
\citet{soumalias2024mlcca,soumalias2025mlhca} replace explicit Bayesian posteriors with per-bidder neural value models but still solve an optimization problem after updating those models.
Our policy instead amortizes the price-update computation into offline training.

Deep learning has also been used to design single-shot auctions~\citep{duetting2019optimal,feng2018deep,rahme2021auction,curry2022differentiable}, and reinforcement learning has been applied to sequential posted-price mechanisms and broader economic design problems~\citep{eden2021reinforcement,gerstgrasser2026stackelberg,ravindranath2024deep}.
Those mechanisms typically visit agents sequentially or optimize revenue.
We study common item prices posted simultaneously to all bidders and use aggregate demand to seek exact market clearing.
A separate line uses multi-agent learning to study bidder strategies under a fixed auction~\citep{deon2024understanding}; here the bidder response rule is fixed and the auctioneer learns.

\section{Experiments}
\label{sec:experiments}

We test whether learned price queries acquire useful information and how learned price selection trades off clearing performance against online computation.
A controlled one-item experiment with three value profiles makes the information mechanism visible: the small state space lets us verify that every trained policy learns an informative opening rather than a greedy one.
We then evaluate on the Combinatorial Auction Test Suite (CATS)~\citep{leyton2000towards}, a standard benchmark for iterative combinatorial auctions that generates instances across four structured domains (Arbitrary, Paths, Regions, and Scheduling), where FastCombo achieves competitive clearing with substantially lower online price-selection computation.
Architecture ablations assess the contributions of the graph encoder and the recurrent state by replacing each component under matched training and evaluation.

Our primary baseline is the Bayesian auction of \citet{brero2019fast}, denoted Bayes\(^{\mathrm{GEN}}\).
At each round, it maintains a posterior over bidder values and selects the price vector maximizing a one-step Bayesian score,
\(\Psi_\lambda(p\mid\mathcal H_t)=\mathbb E[\exp\{-\lambda\, G(p;v)\}\mid\mathcal H_t]\),
where \(G\) is the clearing gap and \(\lambda>0\) is a precision parameter that controls the sharpness of the gap penalty.
At large \(\lambda\), the score approaches exact-clearing probability; at moderate \(\lambda\), it favors prices with small but nonzero gaps, broadening the search.
Throughout our experiments we use \(\lambda=1\), matching the default of \citet{brero2019fast}.
On CATS, posterior updates use Monte Carlo expectation maximization (MCEM) with fitted Gaussian-process priors, requiring repeated sampling and optimization at every round.
The one-item comparator instead receives the exact three-profile distribution and computes the exact posterior, so its 1.85-round cost reflects the one-step objective itself rather than inference error.

Both the one-item and CATS experiments use the recurrent graph actor and the privileged critic described in Section~\ref{sec:policy} and Appendix~\ref{app:policy}.
We train these models separately with PPO and generalized advantage estimation~\citep{schulman2017proximal,schulman2016high}, building on Stable-Baselines3~\citep{raffin2021stable}.
The critic uses private values during training~\citep{pinto2018asymmetric}; these values remain absent from the public actor path.
Evaluation is deterministic and retains every prespecified completed training seed.
Each posted price vector counts as one round.
Evaluation stops at clearing or the query limit.

\subsection{Learning Informative Price Queries}
\label{sec:informative-query}

We test whether FastCombo learns to trade immediate clearing probability for information that helps it clear on the next query.

\paragraph{Setting and training.}
There is one item with scalar price \(p\) and two bidders with joint values \((2,1)\), \((4,3)\), or \((6,5)\), with probabilities \(.25\), \(.35\), and \(.40\).
Training and evaluation use this same distribution, and the Bayesian comparator receives it exactly, uses precision \(\lambda=1\), and has no second stage.
Public value and price ceilings are both 6.
Demand uses \(v_i>p\), breaking ties in favor of no demand.
Numerical clearing uses Equation~\eqref{eq:clearing} with raw tolerance \(10^{-3}\).
Five independent seeds each train for 200,704 transitions with reward \(-1\) per query, discount \(\gamma=1\), and a terminal cap of \(H=3\) rounds with no continuation bootstrap.
Every terminal actor is evaluated without checkpoint selection, and expectations use the exact profile probabilities.
Appendix~\ref{app:toy-training} gives the remaining optimization details.

\paragraph{Learned query savings.}
All five independently trained FastCombo policies clear in 1.65 expected rounds, compared with 1.85 for the Bayesian rule with \(\lambda=1\).
Every learned policy clears all three value profiles within two rounds, and every terminal price satisfies \(G=0\).
The improvement comes from the first query: it reveals which price to post next when immediate clearing fails.

Figure~\ref{fig:toy-learning-measured} shows all ten saved evaluations for each of the five training seeds.

\begin{figure}[H]
\centering
\includegraphics[width=\textwidth]{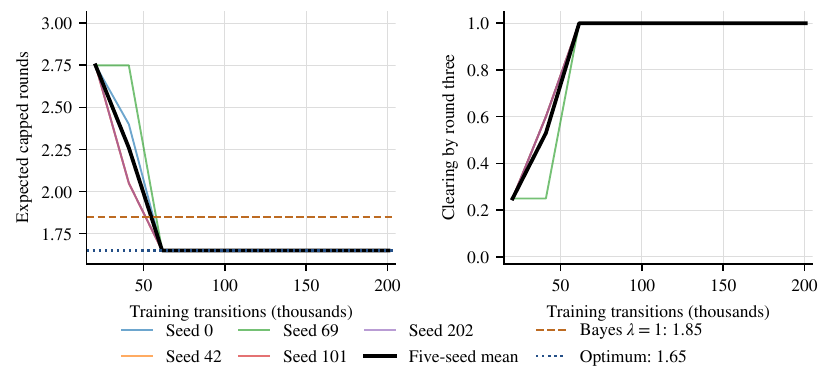}
\caption{Measured training curves from all five original conceptual runs. Left: expected capped rounds under exact profile weights $(.25,.35,.40)$, with failures charged three rounds. Right: clearing probability by round three. Thin lines show individual seeds and the thick line their mean. Dashed references show Bayes $\lambda=1$ (1.85) and the analytical optimum (1.65). Each seed contributes all ten saved evaluations, including its terminal actor.}
\label{fig:toy-learning-measured}
\end{figure}

\paragraph{Interpreting the learned policy.}
The zero-gap price intervals for the low, middle, and high profiles are \([1,2]\), \([3,4]\), and \([5,6]\).
Across the five seeds, learned first prices range from 3.429 to 3.628, clearing the middle profile with probability .35.
Otherwise, zero demanding bidders identifies the low profile and two identifies the high profile.
The respective second prices lie in \([1.479,1.542]\) and \([5.864,6.000]\), clearing the identified profile; all reported prices are rounded to three decimals.

The Bayesian rule instead first posts \(p=5\), clearing the high profile with probability .40.
On failure, both remaining profiles report no demand, so the response does not distinguish them.
The rule next tests the more likely middle profile at \(p=3\), and the low profile requires a third query at \(p=1\).
Greedy clearing tests the profiles in the same order.
FastCombo therefore accepts a five-percentage-point reduction in immediate clearing probability to save .20 rounds in expectation.

\FloatBarrier
\subsection{CATS: Clearing and Online Computation}
\label{sec:cats}

\paragraph{Environment and training.}
We evaluate Arbitrary, Paths, Regions, and Scheduling from the Combinatorial Auction Test Suite (CATS)~\citep{leyton2000towards}, with ten single-minded bidders and twelve items.
Demand follows the weak tie convention in Section~\ref{sec:preliminaries}, and numerical clearing uses Equation~\eqref{eq:clearing} with raw tolerance \(10^{-3}\).
Each domain has ten prespecified seeds, each trained for eight million PPO steps with the original continuous-price FastCombo policy and a cap of \(H=50\) rounds.

Training uses a dense normalized-gap reward to distinguish unsuccessful price vectors.
Writing \(\Delta(\mathcal S,v)=\sum_i v_i-W^\star(\mathcal S,v)=G(0;\mathcal S,v)\), the reward after posting \(p_{t+1}\) is
\begin{equation}
r_t=\begin{cases}
-\max\{0,G(p_{t+1};\mathcal S,v)/\Delta(\mathcal S,v)\},
&\Delta(\mathcal S,v)\geq10^{-8},\\
0,&\text{otherwise}.
\end{cases}
\label{eq:dense-reward-main}
\end{equation}
We optimize discounted return with \(\gamma=.99\); this is a surrogate for the query-count objective in Equation~\eqref{eq:round-reward}.
Clearing terminates training episodes, while unsuccessful time limits use a critic continuation estimate.
We select among 50 checkpoints using clearing rate on a fixed 1,000-instance development population, breaking ties toward the earlier checkpoint, and evaluate on a separate fixed 5,000-instance holdout.
Appendices~\ref{app:training-pomdp} and~\ref{app:training-curves} give the remaining optimization settings and development curves.

\paragraph{Baselines and comparisons.}
Both CATS studies use the default fitted-prior general-distribution rule of \citet{brero2019fast}, Bayes\(^{\mathrm{GEN}}\) (Bayes-GEN in tables), with precision \(\lambda=1\) and the original second stage.
The clearing benchmark also includes subgradient auctions, which adjust prices in response to excess demand with a scalar step-size coefficient.
SG-Auction\(^{D}\) (SG-D) retrospectively selects one coefficient per domain; SG-Auction\(^{I}\) (SG-I) selects it separately for each instance.
Both are optimistic oracles within the tested coefficient family and retain their second stages; their tuning uses evaluation outcomes, whereas FastCombo checkpoint selection uses development auctions.
We report full-holdout clearing rates and compare rounds on instances cleared by every included method.
The linear-price ceiling is the fraction of instances for which an LP dual-optimality test certifies a clearing vector under the same price cap and tolerance.
Paired bootstrap comparisons account for source files and FastCombo training seeds, with baseline fits and coefficient choices fixed; Appendix~\ref{app:baselines} gives the protocol.

\paragraph{Clearing and rounds.}

FastCombo trails Bayes\(^{\mathrm{GEN}}\) by 0.36--2.54 percentage points in clearing across domains (Table~\ref{tab:main-results}).
The paired 95\% confidence interval for the Paths difference includes zero; the intervals for the other three domains favor Bayes\(^{\mathrm{GEN}}\).
FastCombo exceeds the domain-level SG-Auction\(^{D}\) oracle in every domain. It exceeds the instance-level SG-Auction\(^{I}\) oracle on Arbitrary and Regions, but trails it by 1.54 points on Paths and 0.33 on Scheduling.
Across all four domains, FastCombo remains within 3.54 percentage points of the linear-price ceiling.
The clearing and computation comparisons use the same FastCombo policies.

\begin{table}[H]
\centering
\caption{CATS mean clearing (\%) and rounds. Rounds use, for each FastCombo seed, the instances cleared by all four methods. Standard errors, paired intervals, and subset sizes are in Appendix Tables~\ref{tab:clearing-uncertainty} and~\ref{tab:common-cleared-queries}.}
\label{tab:main-results}
\small
\setlength{\tabcolsep}{3pt}
\renewcommand{\arraystretch}{1.05}
\begin{tabular}{@{}l*{4}{rr}@{}}
\toprule
 & \multicolumn{2}{c}{Arbitrary} & \multicolumn{2}{c}{Paths} & \multicolumn{2}{c}{Regions} & \multicolumn{2}{c}{Scheduling} \\
\cmidrule(lr){2-3}\cmidrule(lr){4-5}\cmidrule(lr){6-7}\cmidrule(lr){8-9}
Method & Clear $\uparrow$ & Rounds $\downarrow$ & Clear $\uparrow$ & Rounds $\downarrow$ & Clear $\uparrow$ & Rounds $\downarrow$ & Clear $\uparrow$ & Rounds $\downarrow$ \\
\midrule
FastCombo & 68.74 & 6.59 & 95.10 & 7.95 & 94.23 & 4.80 & 97.59 & 4.64 \\
Bayes-GEN & 71.28 & 3.45 & 95.46 & 4.54 & 95.54 & 3.32 & 98.36 & 4.86 \\
SG-D & 42.46 & 29.78 & 80.38 & 16.93 & 74.86 & 21.99 & 88.28 & 16.32 \\
SG-I & 60.12 & 18.21 & 96.64 & 7.25 & 92.04 & 13.37 & 97.92 & 9.15 \\
\midrule
Linear-price ceiling & 72.28 & \textemdash & 97.48 & \textemdash & 96.22 & \textemdash & 99.90 & \textemdash \\
\bottomrule
\end{tabular}
\end{table}

The round counts in Table~\ref{tab:main-results} use, for each FastCombo seed, the instances cleared by that actor and all three baselines, averaging the ten resulting means.
On these matched success sets, Bayes uses fewer rounds in three domains; FastCombo uses fewer on Scheduling (4.64 versus 4.86).
FastCombo uses fewer rounds than SG-Auction\(^{D}\) in every domain and SG-Auction\(^{I}\) in three.
These counts describe the common-cleared subsets; Appendix Table~\ref{tab:common-cleared-queries} gives their sizes and uncertainty details.

\paragraph{Online computation.}
\label{sec:computation}

The computation study uses 100 auctions per domain and the same assigned CPU within each paired comparison.
It measures online price selection, excluding offline training, observation construction, clearing evaluation, and communication; Appendix~\ref{app:computation} gives the timing protocol.

FastCombo takes 4.97--35.44 milliseconds of measured price-selection computation per auction; the default Bayes\(^{\mathrm{GEN}}\) baseline costs roughly 300--900 times as much (Table~\ref{tab:compute-paired}).
On this 100-auction-per-domain timing sample, FastCombo clears 0.9--2.8 percentage points fewer auctions and uses 0.02--4.30 more capped rounds on average.
These round means include every auction, charging failures all 50 rounds; Table~\ref{tab:main-results} instead reports rounds on common-cleared subsets of the 5,000-instance holdout.

Total auction duration also includes bidder-response delays, which are excluded from the measured price-selection computation.

\begin{table}[H]
\centering
\caption{Default Bayesian comparison on 100 auctions per domain. Entries show paired 95\% intervals from 20,000 joint source-file and FastCombo-seed bootstrap draws. Costs are Bayes/FastCombo; differences are FastCombo minus Bayes; failures cost 50 rounds.}
\label{tab:compute-paired}
\small
\begin{tabular}{lrrr}
\toprule
Domain & Cost ratio & Clearing diff. (pp) & Capped rounds diff. \\
\midrule
Arbitrary & 624.0 [550.1, 696.1] & -0.90 [-2.30, -0.10] & 3.80 [2.68, 5.04] \\
Paths & 304.0 [226.6, 397.2] & -2.10 [-5.00, -0.10] & 4.30 [3.20, 5.54] \\
Regions & 571.7 [497.3, 649.4] & -1.80 [-4.20, -0.20] & 2.34 [1.41, 3.43] \\
Scheduling & 871.1 [687.1, 1064.7] & -2.80 [-6.10, -0.30] & 0.02 [-1.07, 1.29] \\
\bottomrule
\end{tabular}
\end{table}

\subsection{Architecture Ablations}
\label{sec:ablation-results}
Under matched training, replacing the structured encoder with a recurrent dense control reduces clearing by 6.60--18.52 percentage points, while removing recurrence reduces it by 2.68--5.33 points.
All architecture controls use the CATS training budget, normalized-gap reward (Equation~\eqref{eq:dense-reward-main}), discount \(\gamma=.99\), optimization settings, and checkpoint-selection protocol from Section~\ref{sec:cats}.
Table~\ref{tab:architecture-main} gives the corresponding clearing rates.
For both controls, the paired 95\% intervals favor FastCombo in all four domains.
Appendix~\ref{app:ablations} reports the control specification, paired intervals, and remaining ablations; the encoder comparison supports the complete representation under the tested training protocol.

\begin{table}[H]
\centering
\caption{Architecture controls: mean clearing (\%) and rounds. Clearing averages ten seeds, except no-recurrence Paths (nine). Rounds use instances cleared by all three architectures, with nine matched seeds for Paths and ten elsewhere. These success sets differ from the four-method sets in Table~\ref{tab:main-results}, so FastCombo's round means also differ. Clearing retains the original full-holdout seed counts. Round uncertainty is in Table~\ref{tab:architecture-common-rounds}; paired clearing intervals are in Appendix Table~\ref{tab:ablation-effects}.}
\label{tab:architecture-main}
\small
\setlength{\tabcolsep}{3pt}
\renewcommand{\arraystretch}{1.05}
\begin{tabular}{@{}l*{4}{rr}@{}}
\toprule
 & \multicolumn{2}{c}{Arbitrary} & \multicolumn{2}{c}{Paths} & \multicolumn{2}{c}{Regions} & \multicolumn{2}{c}{Scheduling} \\
\cmidrule(lr){2-3}\cmidrule(lr){4-5}\cmidrule(lr){6-7}\cmidrule(lr){8-9}
Method & Clear $\uparrow$ & Rounds $\downarrow$ & Clear $\uparrow$ & Rounds $\downarrow$ & Clear $\uparrow$ & Rounds $\downarrow$ & Clear $\uparrow$ & Rounds $\downarrow$ \\
\midrule
FastCombo & 68.74 & 7.38 & 95.10 & 7.40 & 94.23 & 5.29 & 97.59 & 4.65 \\
Recurrent dense control & 56.13 & 20.95 & 76.58 & 21.77 & 87.64 & 13.47 & 90.80 & 12.63 \\
No recurrence & 63.80 & 9.85 & 89.66 & 8.01 & 91.28 & 6.29 & 94.92 & 5.17 \\
\bottomrule
\end{tabular}
\end{table}

\section{Scope and Limitations}
\label{sec:limitations}

We study single-minded bidders with public bundles and myopic demand.
The simulator's rewards, critic, and clearing test use private values; public clearing certificates and strategyproof allocation and payment rules remain separate mechanism-design requirements~\citep{ausubel2006efficient,mishra2007ascending,gul2000english}.
A finite GRU may lose relevant history, and fixed-size, separately trained CATS domains leave scaling and transfer untested.
The one-item experiment establishes informative learning on a three-profile distribution under the specified Bayesian objective and precision.
It does not establish the same information-gathering mechanism on CATS or a general query advantage across auctions.
The computation measurements exclude offline training and bidder-response costs.

\section{Conclusion}
\label{sec:conclusion}

FastCombo combines training over complete auction trajectories with inexpensive online price evaluation.
Greedy clearing can be arbitrarily suboptimal in expected rounds even under a known true distribution, because information from an unsuccessful query can improve later decisions.
The one-item experiment shows how learned openings use demand to distinguish the remaining value profiles and reduce subsequent queries; CATS demonstrates competitive clearing with substantially lower measured price-selection computation.
Broader round savings and extensions to larger markets, multi-minded and spectrum valuations, and strategic bidders remain open.

\printbibliography

@article{ausubel2006efficient,
  author    = {Ausubel, Lawrence M.},
  title     = {An Efficient Dynamic Auction for Heterogeneous Commodities},
  journal   = {The American Economic Review},
  volume    = {96},
  number    = {3},
  pages     = {602--629},
  year      = {2006}
}

@article{bikhchandani2002package,
  author  = {Bikhchandani, Sushil and Ostroy, Joseph M.},
  title   = {The Package Assignment Model},
  journal = {Journal of Economic Theory},
  volume  = {107},
  number  = {2},
  pages   = {377--406},
  year    = {2002}
}

@inproceedings{cho2014gru,
  author    = {Cho, Kyunghyun and van Merri{\"e}nboer, Bart and Gulcehre, Caglar and Bahdanau, Dzmitry and Bougares, Fethi and Schwenk, Holger and Bengio, Yoshua},
  title     = {Learning Phrase Representations Using {RNN} Encoder--Decoder for Statistical Machine Translation},
  booktitle = {Proceedings of the 2014 Conference on Empirical Methods in Natural Language Processing},
  pages     = {1724--1734},
  year      = {2014}
}

@inproceedings{brero2018bayesian,
  author    = {Brero, Gianluca and Lahaie, S\'{e}bastien},
  title     = {A {B}ayesian Clearing Mechanism for Combinatorial Auctions},
  booktitle = {Proceedings of the 32nd AAAI Conference on Artificial Intelligence},
  year      = {2018}
}

@book{cramton2006combinatorial,
  editor    = {Cramton, Peter and Shoham, Yoav and Steinberg, Richard},
  title     = {Combinatorial Auctions},
  publisher = {MIT Press},
  year      = {2006}
}

@inproceedings{curry2022differentiable,
  author    = {Curry, Michael and Sandholm, Tuomas and Dickerson, John},
  title     = {Differentiable Economics for Randomized Affine Maximizer Auctions},
  booktitle = {Proceedings of the Thirty-Second International Joint Conference on Artificial Intelligence},
  pages     = {2633--2641},
  year      = {2023}
}

@inproceedings{duetting2019optimal,
  author    = {D\"{u}tting, Paul and Feng, Zhe and Narasimhan, Harikrishna and Parkes, David C. and Ravindranath, Sai Srivatsa},
  title     = {Optimal Auctions through Deep Learning},
  booktitle = {Proceedings of the 36th International Conference on Machine Learning},
  volume    = {97},
  pages     = {1706--1715},
  year      = {2019}
}

@inproceedings{feng2018deep,
  author    = {Feng, Zhe and Narasimhan, Harikrishna and Parkes, David C.},
  title     = {Deep Learning for Revenue-Optimal Auctions with Budgets},
  booktitle = {Proceedings of the 17th International Conference on Autonomous Agents and Multiagent Systems},
  year      = {2018}
}

@inproceedings{gilmer2017mpnn,
  author    = {Gilmer, Justin and Schoenholz, Samuel S. and Riley, Patrick F. and Vinyals, Oriol and Dahl, George E.},
  title     = {Neural Message Passing for Quantum Chemistry},
  booktitle = {Proceedings of the 34th International Conference on Machine Learning},
  volume    = {70},
  pages     = {1263--1272},
  year      = {2017}
}

@article{gul2000english,
  author  = {Gul, Faruk and Stacchetti, Ennio},
  title   = {The {E}nglish Auction with Differentiated Commodities},
  journal = {Journal of Economic Theory},
  volume  = {92},
  number  = {1},
  pages   = {66--95},
  year    = {2000}
}

@article{kwasnica2005new,
  author  = {Kwasnica, Anthony M. and Ledyard, John O. and Porter, Dave and DeMartini, Christine},
  title   = {A New and Improved Design for Multiobject Iterative Auctions},
  journal = {Management Science},
  volume  = {51},
  number  = {3},
  pages   = {419--434},
  year    = {2005}
}

@inproceedings{leyton2000towards,
  author    = {Leyton-Brown, Kevin and Pearson, Mark and Shoham, Yoav},
  title     = {Towards a Universal Test Suite for Combinatorial Auction Algorithms},
  booktitle = {Proceedings of the 2nd ACM Conference on Electronic Commerce},
  pages     = {66--76},
  year      = {2000}
}

@article{mishra2007ascending,
  author  = {Mishra, Debasis and Parkes, David C.},
  title   = {Ascending Price {V}ickrey Auctions for General Valuations},
  journal = {Journal of Economic Theory},
  volume  = {132},
  number  = {1},
  pages   = {335--366},
  year    = {2007}
}

@inproceedings{parkes1999ibundle,
  author    = {Parkes, David C.},
  title     = {i{B}undle: An Efficient Ascending Price Bundle Auction},
  booktitle = {Proceedings of the 1st ACM Conference on Electronic Commerce},
  pages     = {148--157},
  year      = {1999}
}

@incollection{parkes2006iterative,
  author    = {Parkes, David C.},
  title     = {Iterative Combinatorial Auctions},
  booktitle = {Combinatorial Auctions},
  editor    = {Cramton, Peter and Shoham, Yoav and Steinberg, Richard},
  publisher = {MIT Press},
  year      = {2006}
}

@inproceedings{pinto2018asymmetric,
  author    = {Pinto, Lerrel and Andrychowicz, Marcin and Welinder, Peter and Zaremba, Wojciech and Abbeel, Pieter},
  title     = {Asymmetric Actor Critic for Image-Based Robot Learning},
  booktitle = {Proceedings of Robotics: Science and Systems},
  year      = {2018}
}

@inproceedings{rahme2021auction,
  author    = {Rahme, Jad and Jelassi, Samy and Bruna, Joan and Weinberg, S. Matthew},
  title     = {A Permutation-Equivariant Neural Network Architecture for Auction Design},
  booktitle = {Proceedings of the 35th AAAI Conference on Artificial Intelligence},
  year      = {2021}
}

@inproceedings{schulman2016high,
  author    = {Schulman, John and Moritz, Philipp and Levine, Sergey and Jordan, Michael and Abbeel, Pieter},
  title     = {High-Dimensional Continuous Control Using Generalized Advantage Estimation},
  booktitle = {Proceedings of the 4th International Conference on Learning Representations},
  year      = {2016}
}

@article{schulman2017proximal,
  author  = {Schulman, John and Wolski, Filip and Dhariwal, Prafulla and Radford, Alec and Klimov, Oleg},
  title   = {Proximal Policy Optimization Algorithms},
  journal = {arXiv preprint arXiv:1707.06347},
  year    = {2017}
}

@article{raffin2021stable,
  author  = {Raffin, Antonin and Hill, Ashley and Gleave, Adam and Kanervisto, Anssi and Ernestus, Maximilian and Dormann, Noah},
  title   = {{Stable-Baselines3}: Reliable Reinforcement Learning Implementations},
  journal = {Journal of Machine Learning Research},
  year    = {2021},
  volume  = {22},
  number  = {268},
  pages   = {1--8},
  url     = {http://jmlr.org/papers/v22/20-1364.html}
}

@article{ausubel2017practical,
  author    = {Ausubel, Lawrence M. and Baranov, Oleg V.},
  title     = {A Practical Guide to the Combinatorial Clock Auction},
  journal   = {The Economic Journal},
  volume    = {127},
  number    = {605},
  pages     = {F334--F350},
  year      = {2017}
}

@inproceedings{brero2019fast,
  author    = {Brero, Gianluca and Lahaie, S\'{e}bastien and Seuken, Sven},
  title     = {Fast Iterative Combinatorial Auctions via {B}ayesian Learning},
  booktitle = {Proceedings of the AAAI Conference on Artificial Intelligence},
  volume    = {33},
  pages     = {1820--1828},
  year      = {2019}
}

@inproceedings{eden2021reinforcement,
  author    = {Brero, Gianluca and Eden, Alon and Gerstgrasser, Matthias and Parkes, David C. and Rheingans-Yoo, Duncan},
  title     = {Reinforcement Learning of Sequential Price Mechanisms},
  booktitle = {Proceedings of the AAAI Conference on Artificial Intelligence},
  volume    = {35},
  pages     = {5219--5227},
  year      = {2021}
}

@article{gerstgrasser2026stackelberg,
  author    = {Gerstgrasser, Matthias and Brero, Gianluca and Eden, Alon and Chakrabarti, Darshan and Lepore, Nicolas and Li, Vincent and Mibuari, Eric and Greenwald, Amy and Parkes, David C.},
  title     = {Stackelberg {POMDP}: Learning to Lead via Reinforcement Learning},
  journal   = {arXiv preprint arXiv:2210.03852},
  year      = {2026},
  doi       = {10.48550/arXiv.2210.03852},
  url       = {https://arxiv.org/abs/2210.03852}
}

@misc{canada2014_700mhz,
  author       = {{Industry Canada}},
  title        = {700~{MHz} Spectrum Auction---Process and Results},
  howpublished = {\url{https://www.canada.ca/en/news/archive/2014/02/700-mhz-spectrum-auction-process-results.html}},
  year         = {2014},
  note         = {Accessed: 2026-09-23}
}

@article{ravindranath2024deep,
  author    = {Ravindranath, Sai Srivatsa and Feng, Zhe and Wang, Di and Zaheer, Manzil and Mehta, Aranyak and Parkes, David C.},
  title     = {Deep Reinforcement Learning for Sequential Combinatorial Auctions},
  journal   = {arXiv preprint arXiv:2407.08022},
  year      = {2024}
}

@inproceedings{deon2024understanding,
  author    = {d'Eon, Greg and Newman, Neil and Leyton-Brown, Kevin},
  title     = {Understanding Iterative Combinatorial Auction Designs via Multi-Agent Reinforcement Learning},
  booktitle = {Proceedings of the 25th ACM Conference on Economics and Computation},
  year      = {2024}
}

@inproceedings{soumalias2024mlcca,
  author    = {Soumalias, Ermis and Weissteiner, Jakob and Heiss, Jakob and Seuken, Sven},
  title     = {Machine Learning-powered Combinatorial Clock Auction},
  booktitle = {Proceedings of the AAAI Conference on Artificial Intelligence},
  volume    = {38},
  pages     = {9891--9900},
  year      = {2024}
}

@inproceedings{soumalias2025mlhca,
  author    = {Soumalias, Ermis and Heiss, Jakob and Weissteiner, Jakob and Seuken, Sven},
  title     = {Prices, Bids, Values: One {ML}-Powered Combinatorial Auction to Rule Them All},
  booktitle = {Proceedings of the 42nd International Conference on Machine Learning},
  volume    = {267},
  pages     = {56570--56614},
  year      = {2025}
}

\clearpage
\appendix
\raggedbottom

\section{Theoretical Details}
\label{app:theory}

\subsection{Clearing Gap and Efficient Allocations}
\label{app:clearing}

For any feasible allocation \(x\), write \(W(x)=\sum_i v_i x_i\).
Then
\begin{align}
L(p;\mathcal S,v)-W(x)
&=\sum_i\left[\max\{0,v_i-p(S_i)\}-x_i\bigl(v_i-p(S_i)\bigr)\right]\notag\\
&\quad+\sum_j p_j\left(1-\sum_{i:j\in S_i}x_i\right)\geq0.
\label{eq:gap-decomposition}
\end{align}
Each bidder term is nonnegative because buying or abstaining cannot yield more than maximum utility; each item term is nonnegative by feasibility and nonnegative prices.
At an efficient allocation, equality holds precisely when every bidder receives a utility-maximizing choice and every unsold item has price zero.
This gives the welfare interpretation of \(G=0\)~\citep{bikhchandani2002package,brero2018bayesian}.
At an indifference price, the allocation can choose among tied utility maximizers even if the demand-report convention selects one of them.
The numerical tolerance in Equation~\eqref{eq:clearing} approximates this equality and does not assert that linear clearing prices exist for every instance.

\subsection{Proof of the Greedy-Clearing Separation}
\label{app:greedy-separation}

\begin{proof}[Proof of Proposition~\ref{thm:greedy-separation}]
Both bidders desire the single item, whose scalar price is \(p\).
Index the possible value profiles by \(k\in\{1,\ldots,K\}\) and set
\begin{equation}
v^{(k)}=\left(\frac{k}{K},\frac{2k-1}{2K}\right),
\qquad
\rho_k=\Pr(v=v^{(k)})=\frac{2k}{K(K+1)}.
\label{eq:greedy-construction}
\end{equation}
These probabilities are positive, strictly increasing, and sum to one.
Prices and values lie in \([0,1]\).
Write \(G_k(p):=G(p;\mathcal S,v^{(k)})\) for the gap at profile \(k\).
For a closed interval \(I\), define \(\operatorname{dist}(p,I):=\min_{z\in I}|p-z|\), the distance to its nearest point.
Substitution in Equation~\eqref{eq:clearing-objective} gives
\[
G_k(p)=\operatorname{dist}(p,I_k),
\qquad
I_k=\left[\frac{2k-1}{2K},\frac{k}{K}\right].
\]
The distance is zero precisely inside the interval.
The intervals are disjoint and ordered, so every price clears at most one profile.

Suppose the remaining profiles are \(1,\ldots,r\), with posterior probabilities proportional to their indices.
Every greedy price belongs to \(I_r\), because profile \(r\) has uniquely greatest posterior probability.
On failure, the true index is below \(r\), and both bidder values are at most \((r-1)/K\), strictly below every price in \(I_r\).
Every remaining profile therefore produces the same response: no demand and failure to clear.
Conditioning removes only profile \(r\), leaving a posterior proportional to the indices \(1,\ldots,r-1\).
Induction implies that every greedy policy tests profiles in the order \(K,K-1,\ldots,1\), regardless of how it chooses or randomizes prices within the selected interval.
Thus profile \(k\) clears in round \(K-k+1\), and
\begin{equation}
J_F(\pi_{\mathrm{greedy}})
=\sum_{k=1}^K\frac{2k(K-k+1)}{K(K+1)}
=\frac{K+2}{3}.
\label{eq:greedy-query-cost}
\end{equation}

For the informative policy, suppose the remaining indices form a consecutive block \(\ell,\ldots,u\).
Choose its middle index \(j=\lfloor(\ell+u)/2\rfloor\) and post the interior price
\(p=(4j-1)/(4K)\in I_j\).
Profile \(j\) clears; lower indices produce zero demands and higher indices produce two demands.
Consequently each unsuccessful query identifies a consecutive block containing at most half as many profiles.
The maximum number of remaining rounds satisfies \(T(0)=0\) and \(T(M)=1+T(\lfloor M/2\rfloor)\) for a block of \(M\geq1\) profiles.
Induction gives \(T(M)=\lceil\log_2(M+1)\rceil\), including the clearing round.
All informative prices are strictly inside clearing intervals, and greedy failures involve strict inequalities, so the argument is unaffected by either demand tie convention when exact-clearing feedback is supplied.

Since the optimal expected round count is at most this policy's worst-case count, the ratio of greedy to optimal expected round counts is at least
\[
\frac{K+2}{3\lceil\log_2(K+1)\rceil},
\]
which diverges as \(K\to\infty\).
\end{proof}

\subsection{Analysis of the One-Item Experiment}
\label{app:toy}

The conceptual comparator maximizes the Bayesian one-step price objective of \citet{brero2019fast},
\begin{equation}
\Psi_\lambda(p\mid\mathcal H_t)
=\mathbb E\!\left[\exp\{-\lambda G(p;v)\}\mid\mathcal H_t\right],
\label{eq:toy-bayes-score}
\end{equation}
Here \(\lambda>0\) controls the gap penalty, \(G(p;v)\) suppresses the fixed bundle profile, and the expectation uses the exact posterior given \(\mathcal H_t\), including demand and nonclearing feedback.
As \(\lambda\to\infty\), the score approaches exact-clearing probability; at finite precision it can differ from the greedy criterion in Section~\ref{sec:informative-theory}.

The following theorem establishes optimal expected rounds for the profiles in Section~\ref{sec:informative-query} and a family containing them.

\begin{theorem}[A prior region favoring an informative query]
\label{thm:prior-region}
Let the three value profiles be \(v^{(k)}=(a+ks+w,a+ks)\), \(k=0,1,2\), where \(a\geq0\) is the lowest value, \(s>0\) the regime spacing, and \(0<w<s\) the within-profile value difference.
Let their probabilities satisfy \(\rho_k>0\) and \(\rho_0+\rho_1+\rho_2=1\).
Use strict demand \(v_i>p\), a clearing indicator for \(G=0\), prices in \([0,P]\) with cap \(P\geq a+2s+w\), and a round cap of at least three.
Fix \(\lambda>0\). If
\begin{equation}
\rho_0\leq\rho_1,\qquad
\rho_1+e^{-\lambda s}\rho_0<\rho_2<\tfrac12,
\label{eq:prior-region}
\end{equation}
then a first price in \((a+s,a+s+w)\), followed by a clearing price for the identified regime, minimizes expected rounds among policies that clear all three regimes.
Its cost is \(2-\rho_1\); exact Bayesian updating under Equation~\eqref{eq:toy-bayes-score} costs \(2-\rho_2+\rho_0\), exceeding it by \(1-2\rho_2>0\).
\end{theorem}

\begin{proof}
Let \(a_k=a+ks\), \(b_k=a_k+w\), and \(I_k=[a_k,b_k]\).
For regime \(k\), write \(G_k(p)=G(p;\mathcal S,v^{(k)})=\operatorname{dist}(p,I_k)\), using the interval distance defined in Appendix~\ref{app:greedy-separation}.
Abbreviate the initial score as \(\Psi_\lambda(p):=\Psi_\lambda(p\mid\mathcal H_0)\).
It is strictly convex between consecutive interval endpoints and monotone outside their range, so a maximizer is an endpoint.
Write \(r=e^{-\lambda s}\) and \(u=e^{-\lambda(s-w)}\); then \(0<r<u<1\).
The inward endpoints satisfy \(\Psi_\lambda(b_0)>\Psi_\lambda(a_0)\) and \(\Psi_\lambda(a_2)>\Psi_\lambda(b_2)\).
Since \(\rho_2>\rho_1\geq\rho_0\),
\[
\Psi_\lambda(b_1)-\Psi_\lambda(a_1)=(\rho_2-\rho_0)(u-r)>0,\qquad
\Psi_\lambda(a_2)-\Psi_\lambda(b_0)=(\rho_2-\rho_0)(1-ru)>0.
\]
The remaining comparison is
\[
\Psi_\lambda(a_2)-\Psi_\lambda(b_1)
=(1-u)(\rho_2-\rho_1-r\rho_0)>0.
\]
Thus Bayesian pricing first chooses \(a_2\).
On failure, both low and middle report no demand, leaving posterior weights proportional to \(\rho_0,\rho_1\).
The middle interval is then optimal because \(\rho_1\geq\rho_0\); either choice has the same expected cost when they are equal.
The resulting cost is \(\rho_2+2\rho_1+3\rho_0=2-\rho_2+\rho_0\).

A first price in \((a_1,b_1)\) clears the middle regime and otherwise identifies low or high from zero or two demands, giving cost \(2-\rho_1\).
Any first price can clear at most one regime.
If it clears high, low and middle remain indistinguishable, so the best possible cost is \(2-\rho_2+\rho_0\); if it clears low, the best cost is \(2-\rho_0+\rho_1\).
If it clears neither extreme nor the middle, the cost is at least two.
These costs exceed \(2-\rho_1\) by \(1-2\rho_2>0\), \(2\rho_1-\rho_0>0\), and at least \(\rho_1>0\), respectively.
A first price clearing the middle cannot improve on \(2-\rho_1\), because every uncleared case requires at least one more query.
Averaging over first prices extends the bound to randomized policies.
\end{proof}

\paragraph{Precision in the conceptual experiment.}
For the experimental profiles, \(a=w=1\), \(s=2\), and \((\rho_0,\rho_1,\rho_2)=(.25,.35,.40)\).
Condition~\eqref{eq:prior-region} reduces to \(.35+.25e^{-2\lambda}<.40\), equivalently \(\lambda>\tfrac12\log5\approx.805\).
Throughout this range, including the tested \(\lambda=1\), Bayesian pricing uses \(.40+2(.35)+3(.25)=1.85\) expected rounds, while the informative optimum is 1.65.
At the smaller precision \(\lambda=.1\), the Bayesian rule instead first posts \(p=4\), distinguishes the remaining profiles on failure, and also attains 1.65 expected rounds.

\section{Implementation and Training}
\label{app:policy-details}
\label{app:policy}

\subsection{Public Features and Architecture}

Let \(\ell_{i,t}\) and \(u_{i,t}\) be the public lower and upper bounds on \(v_i\) after \(t\) rounds.
After observing \(y_{i,t+1}\), these bounds update as
\begin{align}
\ell_{i,t+1}&=\begin{cases}
\max\{\ell_{i,t},p_{t+1}(S_i)\},&y_{i,t+1}=1,\\
\ell_{i,t},&y_{i,t+1}=0,
\end{cases}\label{eq:training-lower-bound}\\
u_{i,t+1}&=\begin{cases}
u_{i,t},&y_{i,t+1}=1,\\
\min\{u_{i,t},p_{t+1}(S_i)\},&y_{i,t+1}=0.
\end{cases}\label{eq:training-upper-bound}
\end{align}
Initially \(\ell_{i,0}=0\), \(u_{i,0}=C_d\), and, for every item \(j\), \(p_{0,j}=D_{0,j}=0\) and \(h_{0,j}=0\); these placeholders do not represent an executed query.
For the strict-demand conceptual experiment, these updates retain closed outer bounds at the endpoints.
Bounds, prices, and demand are divided by \(C_d\), \(p_d^{\max}\), and \(n\), respectively, then clipped to \([0,1]\).
All public features are determined by the history in Equation~\eqref{eq:public-history}.
The GRU is shared across items and has 32 hidden units.

Let \(s_i=|S_i|/m\) be normalized bundle size and \(c_j=|\{i:j\in S_i\}|\) the number of bundles containing item \(j\).
Write \(b_{i,t}^{(r)}\) and \(e_{j,t}^{(r)}\) for the 128-dimensional bidder and item tokens after \(r\) graph passes.
The initial tokens, with \(r=0\), are
\begin{equation}
\begin{aligned}
b_{i,t}^{(0)}&=\operatorname{ReLU}\!\left(W_{\mathrm{bid}}[\bar\ell_{i,t};\bar u_{i,t};s_i]+\beta_{\mathrm{bid}}\right),\\
e_{j,t}^{(0)}&=\operatorname{ReLU}\!\left(W_{\mathrm{item}}[h_{t,j};t/H;c_j/n]+\beta_{\mathrm{item}}\right).
\end{aligned}
\label{eq:token-initialization}
\end{equation}
Here \([\,\cdot\,;\,\cdot\,]\) denotes vector concatenation, the \(W\)'s and \(\beta\)'s are learned weight matrices and bias vectors, and \(\operatorname{ReLU}(z)=\max\{0,z\}\) acts coordinatewise.
Current prices and aggregate demand enter the item encoder through the recurrent state.
The graph's two passes each concatenate a node's 128-dimensional token with its 128-dimensional neighbor mean and apply a learned map back to 128 dimensions.
Items update first and bidders then use the new item tokens.
The passes have separate weights; within each pass, item maps are shared across items and bidder maps across bidders.
Empty neighborhoods contribute a zero mean while retaining the node's own token.

The actor applies a shared \(128\to128\to64\to1\) decoder with ReLU hidden layers to each final item token, yielding \(\mu_{t,j}\).
It has no separate pooled-context, item-coordinate, or domain-label input.
During training, the unbounded normalized action \(\tilde a_{t,j}\) and its clipped version \(a_{t,j}\) satisfy
\begin{equation}
\tilde a_{t,j}\sim\mathcal N(\mu_{t,j},\sigma_j^2),\qquad
\sigma_j=\exp(\operatorname{clip}_{[-20,5]}(\omega_j)),\qquad
a_{t,j}=\operatorname{clip}_{[-1,1]}(\tilde a_{t,j}),
\label{eq:fastcombo-policy}
\end{equation}
Here \(\mathcal N(\mu,\sigma^2)\) is a Normal distribution with mean \(\mu\) and variance \(\sigma^2\); \(\omega_j\) is a learned state-independent log-scale and \(\sigma_j\) the resulting standard deviation.
The clipped action has a Normal density in \((-1,1)\) and point masses at \(-1\) and \(1\) equal to the corresponding tail probabilities.
PPO scores this censored law, and the posted price is
\(p_{t+1,j}=p_d^{\max}(a_{t,j}+1)/2\).
Evaluation uses the clipped location as in Equation~\eqref{eq:deterministic-price-main}, which generally differs from the censored law's expectation.

For the critic, a learned embedding of \(\bar v_i=v_i/C_d\) is added to each bidder's initial token before a separate call to the shared graph encoder.
The critic concatenates the mean bidder token with all item tokens into \(z_t^{\mathrm{crit}}\in\mathbb R^{128(m+1)}\).
Hidden layers of widths 128 and 64 with hyperbolic-tangent (Tanh) activations predict the expected discounted sum of future training rewards.
Write this estimate as \(V_\phi(\mathcal H_t,v)\), where \(\phi\) includes shared and critic-specific parameters.
Private values never enter the recurrent state or the public actor call.
Figure~\ref{fig:fastcombo-architecture} summarizes these paths.

\begingroup
\renewcommand{\small}{\fontsize{9}{10}\selectfont}
\begin{figure}[t]
\centering
\begin{tikzpicture}[
  x=1cm,
  y=1cm,
  font=\small,
  flow/.style={-{Stealth[length=2.0mm]},semithick,draw=black},
  stem/.style={semithick,draw=black},
  box/.style={draw=black,fill=white,rounded corners=2pt,
    align=center,text width=5.90cm,inner sep=4pt},
  stage/.style={font=\small\bfseries,text=black,align=center},
  bidder/.style={circle,draw=black,fill=white,
    minimum size=4.2mm,inner sep=0pt},
  item/.style={rectangle,draw=black,fill=white,
    minimum size=4.2mm,inner sep=0pt}
]

% Public bidder and recurrent item lanes.
\node[stage] at (3.00,10.80) {1. Inputs and node vectors};
\node[box,minimum height=1.45cm] (btoken) at (3.00,9.25)
  {Bidder \(i\): normalized bounds and bundle size\\
   \((\bar\ell_{i,t},\bar u_{i,t},s_i)\)\\
   Bidder encoder reused for every bidder\\
   \(3\to128:\ b_{i,t}^{(0)}\in\mathbb{R}^{128}\)};
\node[box,minimum height=2.05cm] (itoken) at (3.00,6.85)
  {Item \(j\): current feedback \((\bar p_{t,j},\bar D_{t,j})\)\\
   and previous state \(h_{t-1,j}\)\\
   Shared per-item GRU: \(2\to32:\ h_{t,j}\)\\[2pt]
   Item encoder reused for every item\\
   \([h_{t,j};t/H;c_j/n]\to e_{j,t}^{(0)}\in\mathbb{R}^{128}\)};
\node[box,densely dashed,text width=5.90cm,minimum height=0.65cm]
  (private) at (3.00,5.20) {Private values \(v_i\): critic only};

% Shared graph encoder.
\node[box,minimum height=4.80cm] (gnn) at (9.95,7.70) {};
\node[stage] at (9.95,9.75) {2. Bundle-graph updates};
\node[bidder] (b1) at (8.60,9.05) {};
\node[bidder] (b2) at (8.60,8.60) {};
\node[bidder] (b3) at (8.60,8.15) {};
\node[item] (j1) at (11.10,9.05) {};
\node[item] (j2) at (11.10,8.60) {};
\node[item] (j3) at (11.10,8.15) {};
\draw[thin,black] (b1) -- (j1);
\draw[thin,black] (b1) -- (j2);
\draw[thin,black] (b2) -- (j1);
\draw[thin,black] (b2) -- (j3);
\draw[thin,black] (b3) -- (j2);
\draw[thin,black] (b3) -- (j3);
\node at (8.60,7.72) {bidders};
\node at (11.10,7.72) {items};
\node[align=center,text width=5.50cm] at (9.95,6.82)
  {\textbf{Two untied graph passes}\\
   bidder means \(\longrightarrow\) items\\
   item means \(\longrightarrow\) bidders\\
   shared public/critic weights};
\node[align=center,text width=5.50cm] at (9.95,5.75)
  {Public bundle graph: edge \((i,j)\)\\exactly when \(j\in S_i\)};
\draw[flow] (btoken.east) -- ([yshift=1.35cm]gnn.west);
\draw[flow] (itoken.east) -- ([yshift=-0.85cm]gnn.west);
\draw[flow,densely dashed] (private.east) -- (6.45,5.20)
  |- ([yshift=-1.85cm]gnn.west);

% Readout and output heads continue clockwise.
\node[box,minimum height=3.55cm] (graphreadout) at (9.95,2.65)
  {\textbf{3. Readout}\\[7pt]
   Final graph vectors\\[5pt]
   Public item vectors\\
   \(e_{j,t}^{(2)}\in\mathbb{R}^{128}\): one vector per item\\[7pt]
   Privileged critic input\\
   \(z_t^{\mathrm{crit}}\in\mathbb{R}^{128(m+1)}\)};
\draw[flow] (gnn.south) -- (graphreadout.north);

\node[stage] at (3.00,4.60) {4. Outputs};
\node[box,minimum height=2.80cm] (ahead) at (3.00,2.85)
  {\textbf{Actor reused for every item}\\[2pt]
   \(e_{j,t}^{(2)}\to128\to64\to1\)\\
   output \(\mu_{t,j}\): no pooled context\\
   censored Normal (training)\\
   clipped location (evaluation)\\
   fixed map \(\to p_{t+1}\in[0,p_d^{\max}]^m\)};
\node[box,densely dashed,minimum height=1.70cm] (vhead) at (3.00,0.30)
  {\textbf{Critic (training only)}\\[2pt]
   \(z_t^{\mathrm{crit}}\to128\to64\to1\)\\
   Tanh hidden layers; scalar value estimate\\
   \(V_{\phi}(\mathcal H_t,v)\)};
\coordinate (fork) at (6.55,2.65);
\draw[stem] (graphreadout.west) -- (fork);
\fill[black] (fork) circle[radius=1.35pt];
\draw[flow] (fork) |- (ahead.east);
\draw[flow,densely dashed] (fork) |- (vhead.east);

\end{tikzpicture}
\caption{FastCombo architecture, following the numbered stages clockwise. Solid arrows show the public actor: per-item recurrence, two graph passes, shared itemwise decoding, and the fixed price map. Dashed arrows show the training-only privileged critic, which embeds normalized private values before a separate call to the same graph encoder; \(z_t^{\mathrm{crit}}\) denotes its concatenated graph representation. During training, actor locations and learned state-independent log-scales parameterize a product censored Normal law on \([-1,1]^m\); evaluation uses the clipped location.}
\label{fig:fastcombo-architecture}
\end{figure}
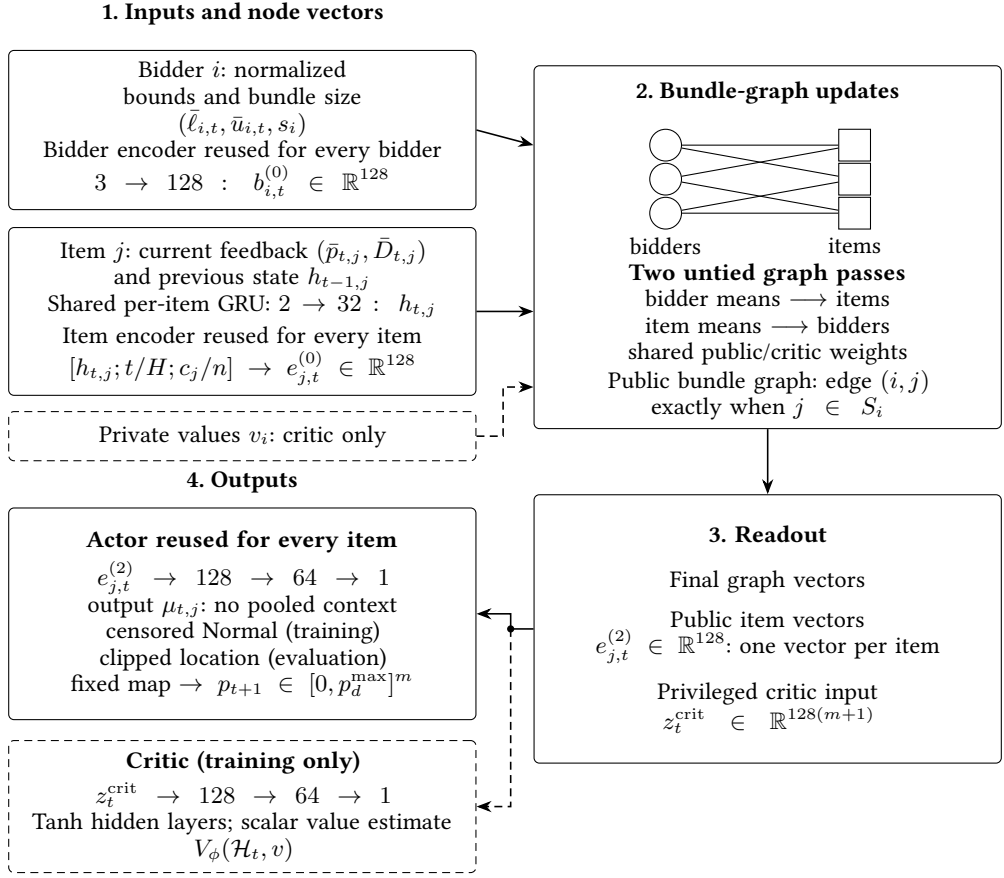

\endgroup

\subsection{Optimization and Query Caps}
\label{app:training-pomdp}

A clearing query is a true termination, including when it occurs at the cap \(H\).
An uncleared CATS cap is a truncation: PPO adds \(\gamma V_\phi\) at the final pre-reset observation to the last reward before advantage estimation, using the updated recurrent state and training-only private values.
This continuation estimate is part of the training surrogate, while evaluation still stops at the cap.

For CATS, each PPO update uses 2,048 transitions: 256 from each of eight environments.
Complete within-rollout episode segments are packed into minibatches of at most 256 transitions and replayed in temporal order.
Episode boundaries reset memory; rollout boundaries preserve its forward value while stopping gradients through earlier rollouts.
CATS uses ten optimization epochs per rollout, Adam learning rate \(2\times10^{-4}\), generalized-advantage parameter .95, PPO clip .2, value coefficient .5, zero entropy coefficient, and gradient-norm clipping at .5.

\subsection{Training Dynamics on CATS}
\label{app:training-curves}

Figure~\ref{fig:learning-exact-clearing} reports clearing at all 50 development evaluations for each of the ten training seeds.
Mean clearing rises rapidly and then remains near a plateau in all four domains, with temporary declines at some individual checkpoints.
The curves show both the variation across training seeds and the role of checkpoint selection under the protocol in Section~\ref{sec:cats}.

\begin{figure}[!t]
  \centering
  \includegraphics[width=\textwidth]{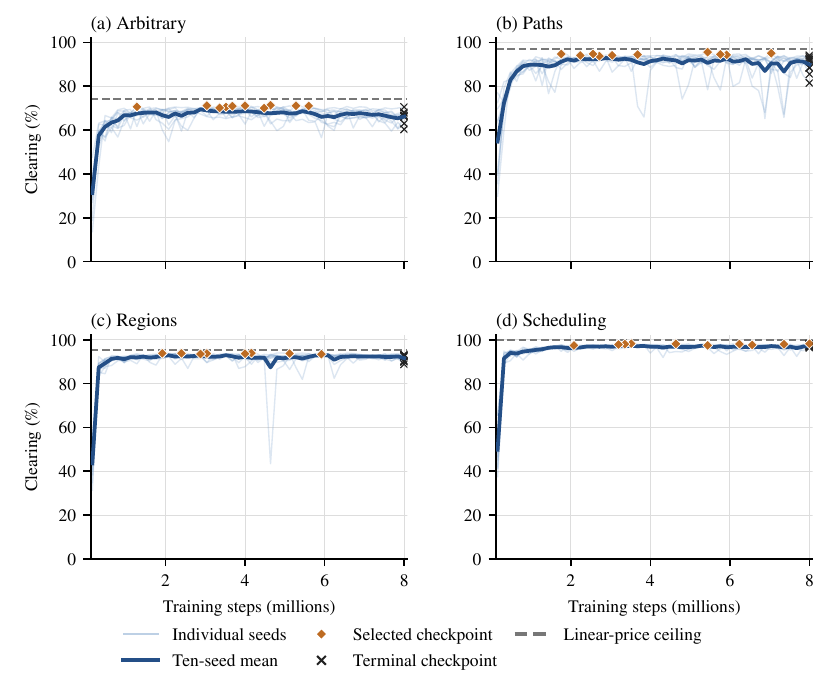}
  \caption{Development clearing over CATS training on the fixed 1,000-instance population. The plotted ``exact clearing'' metric uses the numerical criterion \(G<10^{-3}\). Thin lines show all ten training seeds and the thick line their mean. Diamonds mark the selected checkpoints, crosses mark terminal checkpoints, and the dashed line is the development-population linear-price ceiling. Selection uses development clearing only; final results use a separate holdout.}
  \label{fig:learning-exact-clearing}
\end{figure}

\subsection{Conceptual Experiment Protocol}
\label{app:toy-training}

The one-item experiment in Section~\ref{sec:informative-query} uses the same graph actor, recurrence, continuous-price head, and privileged critic as CATS.
The five fits use seeds \(0,42,69,101,202\), with eight parallel environments per fit.
A fresh profile is drawn at each reset.
The optimizer settings match CATS: ten epochs per rollout, Adam learning rate \(2\times10^{-4}\), GAE parameter .95, 2,048 transitions across eight environments, minibatches of at most 256 transitions, PPO clip .2, value coefficient .5, zero entropy coefficient, and gradient-norm clipping at .5.
Neither experiment uses a target-KL stopping threshold.
The conceptual experiment uses \(\gamma=1\), reward \(-1\) per query, a terminal three-query horizon, and 200,704 training transitions; evaluation uses every terminal actor without checkpoint selection.
CATS instead uses \(\gamma=.99\), the normalized-gap reward, a truncating 50-query cap, and development-selected checkpoints; its eight-million-step budget completes at 8,001,536 transitions because updates use complete rollouts.
Training uses no teacher actions, and evaluation applies no price repair.
The five conceptual-experiment configurations and their source records are supplied in the reproduction package; \path{docs/EXPERIMENTS.md} locates the original runs, and \path{reference/verification/toy_verified.json} records the configuration and outcome checks.

\section{Evaluation Protocols}
\label{app:evaluation}

\subsection{CATS Instances, Baselines, and Uncertainty}
\label{app:baselines}

Section~\ref{sec:cats} defines the domains, splits, training budget, and selection rule.
The ten prespecified training seeds are \(0,42,69,101,202,303,404,505,606,707\).
CATS 2.1 generates the source files; per-file parameters, generation seeds, hashes, and domain-specific valuation and price ceilings are indexed in the reproduction package's \path{manifests/data.json}.
Training, development, and holdout auctions draw from the same 100 source files per domain.
The fixed development and holdout populations contain 1,000 and 5,000 auctions per domain, respectively, with no identical unordered sets of bundle--value pairs between them.
Each holdout prior excludes that auction's ten source bid records.
Thus development--holdout separation is defined at the auction level while retaining a common source-file pool.
Training software inventories cover all 320 runs (40 main-benchmark, 240 ablation, and 40 dense-policy runs) and record Python 3.12.3, PyTorch 2.12.1, and NumPy 2.5.0.
Package documentation in \path{docs/DATA.md} and \path{docs/ENVIRONMENT.md} specifies the sampling procedures, software inventories, and environment used for reproduction.

Bayes\(^{\mathrm{GEN}}\) uses 128 MCEM samples and at most 30 iterations.
The implementation parameter \texttt{accuracy}=7 specifies the sample count \(\min(2^7,128)=128\); the evaluation wrapper passes this count directly.
Both CATS studies scale each source file's maximum positive value to 10 and apply \(\lambda=1\) to the unnormalized gap in those scaled units, without division by total valuation. The clearing tolerance remains equivalent to \(10^{-3}\) in original CATS units.
MCEM starts from a constant price vector \(2\bar\mu/m\), where \(\bar\mu\) is the mean bidder prior mean and \(m\) the item count. With \(e_t=\max_j|p_{t,j}-p_{t-1,j}|/(|p_{t,j}|+10^{-100})\), the native early stop, after at least three iterations, requires \(\max(e_{t-2:t})/(\min(e_{t-2:t})+10^{-100})<2\). The code also checks that the minimum is no larger than the current error.
Both CATS comparisons use the same prior-sampling and bidder-construction procedures.
Monte Carlo seeds were recorded; Gaussian-process optimizer restarts were unseeded in the 5,000-instance benchmark.
For both SG tuning oracles, the query cap applies to the selected auction trajectory; counterfactual runs used to select its coefficient are not charged to that trajectory.
Their restriction to the tested coefficient family means neither is an upper bound on other price policies.
\paragraph{Price updates and query accounting.}
Both baselines query a second price vector after a nonclearing quote with no overdemand: prices of undemanded items are set to zero, and the additional query is charged if budget remains. Bayes conditions its EP updates on both responses. Following an unsuccessful second stage, SG continues from its original first-stage prices and demand. FastCombo has no second stage.
For grid index \(k\in\{0,\ldots,99\}\), SG uses \(a_k=(k+1)\max_i v_i/100\). Writing \(d_j\) for the number of bidders demanding item \(j\), the update is
\[p\leftarrow\max\{0,p+a_k(d-\mathbf 1)/\sqrt{t}\},\qquad t=1+\text{queries already charged}.\]
Initially \(p=d=0\), so the first charged quote is zero; second-stage queries advance \(t\). SG-D maximizes clearing count, then minimizes mean rounds among successful runs, then chooses the smallest index. The selected fractions of maximum value are .08, .17, .05, and .07 for Arbitrary, Paths, Regions, and Scheduling. SG-I chooses a successful trajectory with the fewest queries, breaking ties by smallest index; if none succeeds, it records failure at 50 using the last grid candidate as fallback. Baseline prices are nonnegative but are not clipped to the actor's calibrated price cap.

The prior samples half the positive valid source bids, rounded down, without replacement after excluding the auction's ten source records. The GP uses a linear kernel with amplitude initialized to \(e^2\), bounds \([10^{-3},10^3]\), zero fixed offset, and additive white-noise variance initialized to \(e^{-2}\), bounds \([10^{-5},10^3]\). It adds numerical jitter \(10^{-8}\), does not normalize targets, and performs an initial optimization plus one randomized restart. EP uses \(\beta=10/\bar\mu\) if \(\bar\mu>0\), otherwise 10.

\paragraph{Linear-price ceiling and implementation.}
For the actor's linear-price ceiling, the LP minimizes \(\sum_jp_j+\sum_i u_i\), subject to \(p(S_i)+u_i\geq v_i\), \(u_i\geq0\), and \(0\leq p_j\leq p_d^{\max}\). An optimum less than \(10^{-3}\) above exact integer welfare counts as feasible; the corresponding uncapped LP measures feasibility without the actor's price constraint. Enumerating all 1,024 bidder subsets gives exact integer welfare. The reproduction audit independently recomputes welfare and both LP classifications.
The reproduction package documents the evaluation entrypoints in \path{docs/EXPERIMENTS.md}, including the second-stage wrappers, MCEM updates, GP fitting, and SG coefficient selection. Numerical checks are supplied in \path{reference/verification/baseline_verified.json}.

\paragraph{Statistical comparisons.}
Clearing-rate standard errors and paired intervals in Table~\ref{tab:clearing-uncertainty} use 20,000 bootstrap resamples of source files, drawn jointly across methods, and FastCombo training seeds.
Baseline fits and coefficient choices stay fixed, so their uncertainty is conditional on those choices and the recorded solver outcomes.
The paired FastCombo--Bayes interval uses the 2.5th and 97.5th percentiles of the same bootstrap differences.
Ablation intervals additionally pair the matched training seeds.
Query comparisons on common-cleared subsets condition on all included methods succeeding and therefore describe a selected subset of instances.
Table~\ref{tab:common-cleared-queries} reports these comparisons; failure rates remain visible in Table~\ref{tab:main-results}.

\begin{table}[H]
\centering
\caption{CATS clearing (\%; mean $\pm$ bootstrap standard error) on the 5,000-instance holdout. The last column gives the paired FastCombo--Bayes difference in percentage points with a pointwise 95\% confidence interval. FC denotes FastCombo; CI denotes confidence interval.}
\label{tab:clearing-uncertainty}
\footnotesize
\setlength{\tabcolsep}{3pt}
\renewcommand{\arraystretch}{1.05}
\begin{tabular}{@{}lrrrrrr@{}}
\toprule
Domain & FastCombo & Bayes-GEN & SG-D & SG-I & Ceiling & FC$-$Bayes [95\% CI] \\
\midrule
Arbitrary & 68.74 $\pm$ 0.96 & 71.28 $\pm$ 0.93 & 42.46 $\pm$ 0.98 & 60.12 $\pm$ 1.04 & 72.28 & -2.54 [-3.03, -2.10] \\
Paths & 95.10 $\pm$ 0.42 & 95.46 $\pm$ 0.37 & 80.38 $\pm$ 0.71 & 96.64 $\pm$ 0.35 & 97.48 & -0.36 [-0.87, +0.15] \\
Regions & 94.23 $\pm$ 0.40 & 95.54 $\pm$ 0.34 & 74.86 $\pm$ 0.81 & 92.04 $\pm$ 0.46 & 96.22 & -1.31 [-1.68, -0.95] \\
Scheduling & 97.59 $\pm$ 0.24 & 98.36 $\pm$ 0.20 & 88.28 $\pm$ 0.52 & 97.92 $\pm$ 0.22 & 99.90 & -0.77 [-1.16, -0.39] \\
\bottomrule
\end{tabular}
\end{table}

\begin{table}[H]
\centering
\caption{Mean rounds on seed-specific common-cleared subsets, with one charged demand query per round. For each FastCombo seed, the subset contains instances cleared by that actor and all three baselines; entries average the ten seed-specific means. Entries are means $\pm$ bootstrap standard errors. Bayes-GEN denotes Bayes\(^{\mathrm{GEN}}\); SG-D and SG-I use domain-level and instance-level oracle tuning, respectively. Standard errors use 20,000 paired source-file and training-seed bootstrap resamples, recomputing each selected seed's common-set mean within each resample. Baselines use the same selected seeds and source files; their errors therefore include holdout variation, not only variation in subset membership. The last column reports the mean subset size and its range.}
\label{tab:common-cleared-queries}
\footnotesize
\setlength{\tabcolsep}{3pt}
\begin{tabular}{@{}lrrrrr@{}}
\toprule
Domain & FastCombo & Bayes-GEN & SG-D & SG-I & Common set size [range] \\
\midrule
Arbitrary & 6.59 $\pm$ 0.18 & 3.45 $\pm$ 0.07 & 29.78 $\pm$ 0.26 & 18.21 $\pm$ 0.27 & 2113.6 [2108--2118] \\
Paths & 7.95 $\pm$ 0.25 & 4.54 $\pm$ 0.12 & 16.93 $\pm$ 0.23 & 7.25 $\pm$ 0.12 & 3964.8 [3952--3975] \\
Regions & 4.80 $\pm$ 0.10 & 3.32 $\pm$ 0.05 & 21.99 $\pm$ 0.20 & 13.37 $\pm$ 0.17 & 3735.2 [3727--3740] \\
Scheduling & 4.64 $\pm$ 0.13 & 4.86 $\pm$ 0.06 & 16.32 $\pm$ 0.19 & 9.15 $\pm$ 0.13 & 4394.5 [4384--4403] \\
\bottomrule
\end{tabular}
\end{table}

\subsection{Online Computation}
\label{app:computation}

The computation study selects one auction uniformly from each of the 100 source files per domain, using seed 20260911 on the saved holdout.
It retains all 40 final actors and the benchmark's tolerance, round cap, and Bayesian prior samples.
Each Gaussian process is refitted with optimizer restarts seeded by the recorded Bayesian seed and then held fixed across configurations and repetitions.
The runtime comparison uses the recorded default Bayes\(^{\mathrm{GEN}}\) configuration: 128 MCEM samples, at most 30 iterations, \(\lambda=1\), native early convergence, and the charged second stage.
SG-Auction\(^{D}\) selects its coefficient on the other 4,900 holdout auctions per domain, excluding the timing sample.
Its coefficient is fixed before the timed evaluations; it has no oracle access to outcomes on those 100 measurement instances.
The timing study is therefore a separate 400-auction comparison.

Each auction receives an untimed reference execution and three timed repetitions, with randomized method order.
Methods run sequentially on the same assigned CPU, using one Torch/BLAS thread without batching.
The 40 ten-auction tasks used AMD EPYC 7543 or Intel Xeon Gold 5118, 6148, or 6238R processors; comparisons within an auction use the same processor.
Measured components are FastCombo prediction and price conversion, Bayesian MCEM price generation and expectation-propagation updates, and subgradient price arithmetic.
Every charged query remains included, including second-stage queries without an MCEM call.
Observation construction, clearing checks, loop bookkeeping, communication, Gaussian-process setup, model loading, and offline training are excluded.
Instrumented executions reproduce the uninstrumented outcomes and price traces on this population.

Cumulative component cost is the median over timing repetitions for each auction and actor; domain means average auctions and all ten FastCombo seeds, including failures.
The paired uncertainty procedure uses 20,000 joint bootstrap draws of source files and FastCombo seeds, with percentile 95\% intervals in Table~\ref{tab:compute-paired}. It conditions on the recorded fitted priors and median-of-three timings; timing repetitions are not resampled as independent auctions.
The 40 full training fits used 3,812.54 allocated CPU-hours including their in-job evaluation, with eight CPUs per fit and domain-average elapsed times of 10.05--13.67 hours.
These totals exclude earlier architecture and hyperparameter searches, do not measure utilization, and use hardware different from the timing study, so they do not determine a same-hardware break-even count.

\section{Architecture Controls and Ablations}
\label{app:ablations}

We evaluate five component ablations and a recurrent non-graph control under the training and evaluation protocol of Section~\ref{sec:cats}.
The control replaces the graph encoder with a shared dense network receiving the complete public observation, all item recurrent states, the local item state, and an item identifier.
It retains the twelve 32-dimensional item states updated by a shared GRU, 128-dimensional item tokens, price-head and critic-head architectures, and training-only access to private values.
Graph and dense models have 386,446 and 386,022 total training parameters (a 0.11\% difference); their stochastic actor paths contain 131,981 and 162,721 parameters.
The dense control was chosen among five recurrent candidates using development auctions only, then trained afresh for ten seeds per domain.

This comparison supports the complete structured encoder under the shared PPO protocol; it does not isolate message passing from parameter sharing, representation layout, or optimization effects.
Changing graph depth or the critic has smaller, domain-dependent effects, and removing bidder bounds changes clearing by less than half a percentage point.

Table~\ref{tab:ablation-effects} reports paired changes in clearing; the final architecture's absolute rates are in Table~\ref{tab:main-results}.
No-recurrence Paths uses nine matched seeds because one prespecified ablation run did not complete; all other cells use ten.
The non-graph comparison re-evaluates the final actors on the same instances; its Paths reference is 95.106\%, versus 95.104\% in the original benchmark (one clearance among 50,000 evaluations).
Paired differences use this re-evaluated reference; absolute final-architecture rates retain the original benchmark values.

The round comparisons use the original FastCombo benchmark outcomes, with matched seeds and the seed-specific intersection of successes across all three architectures. Mean common-set sizes (minimum--maximum over seeds) are 2,701.5 (2,460--2,848), 3,752.2 (3,575--3,927), 4,299.1 (4,171--4,381), and 4,484.5 (4,367--4,553) for Arbitrary, Paths, Regions, and Scheduling, respectively. Tables~\ref{tab:architecture-common-rounds} and~\ref{tab:architecture-capped-rounds} give 20,000-draw paired source-file/seed bootstrap intervals. The capped means use all matched-seed auctions and assign failures 50 rounds. Full-holdout clearing columns retain their original seed counts; Paths rounds use nine seeds for every architecture.

\begin{table}[H]
\centering
\caption{Paired change in clearing at raw gap tolerance $10^{-3}$ relative to the final depth-2 architecture (alternative minus final, in percentage points), with pointwise 95\% paired source-file and training-seed bootstrap intervals. Graph depth is the number of graph passes. Positive values favor the alternative. References and seed accounting follow Section~\ref{app:ablations}.}
\label{tab:ablation-effects}
\footnotesize
\setlength{\tabcolsep}{2.5pt}
\begin{tabular}{@{}lrrrr@{}}
\toprule
Alternative & Arbitrary & Paths & Regions & Scheduling \\
\addlinespace[2pt]
\midrule
No recurrence & -4.94 [-5.74, -4.15] & -5.33 [-6.13, -4.56] & -2.95 [-3.45, -2.49] & -2.68 [-3.19, -2.19] \\
\addlinespace[2pt]
No bidder bounds & +0.26 [-0.26, +0.80] & +0.36 [+0.03, +0.70] & +0.16 [-0.16, +0.52] & +0.45 [+0.09, +0.82] \\
\addlinespace[2pt]
\begin{tabular}[c]{@{}l@{}}Recurrent non-graph\\encoder\end{tabular} & \shortstack[r]{-12.61\\{}[-14.42, -10.96]} & \shortstack[r]{-18.52\\{}[-20.56, -16.48]} & \shortstack[r]{-6.60\\{}[-7.62, -5.63]} & \shortstack[r]{-6.79\\{}[-7.71, -5.89]} \\
\addlinespace[2pt]
\begin{tabular}[c]{@{}l@{}}One graph pass\\(depth 1)\end{tabular} & -0.94 [-1.45, -0.48] & -0.64 [-1.11, -0.16] & -0.12 [-0.49, +0.23] & +0.31 [+0.00, +0.66] \\
\addlinespace[2pt]
\begin{tabular}[c]{@{}l@{}}Three graph passes\\(depth 3)\end{tabular} & -1.06 [-1.54, -0.58] & -0.10 [-0.57, +0.34] & -0.65 [-1.02, -0.27] & -0.46 [-0.77, -0.17] \\
\addlinespace[2pt]
\begin{tabular}[c]{@{}l@{}}Public-information\\critic\end{tabular} & -0.61 [-1.19, -0.02] & -0.14 [-0.62, +0.33] & +0.05 [-0.21, +0.33] & +0.18 [-0.09, +0.47] \\
\addlinespace[2pt]
\bottomrule
\end{tabular}
\end{table}

\begin{table}[H]
\centering
\caption{Rounds on seed-specific three-method common-cleared subsets. Entries show mean and paired percentile 95\% bootstrap interval, using 20,000 source-file and matched-seed draws.}
\label{tab:architecture-common-rounds}
\small
\begin{tabular}{lrrr}
\toprule
Domain & FastCombo & Dense & No recurrence \\
\midrule
Arbitrary & 7.38 [7.08, 7.69] & 20.95 [19.23, 22.56] & 9.85 [9.09, 10.61] \\
Paths & 7.40 [6.97, 7.90] & 21.77 [20.01, 23.24] & 8.01 [7.59, 8.50] \\
Regions & 5.29 [5.07, 5.51] & 13.47 [12.71, 14.23] & 6.29 [5.99, 6.58] \\
Scheduling & 4.65 [4.42, 4.91] & 12.63 [11.82, 13.51] & 5.17 [4.97, 5.36] \\
\bottomrule
\end{tabular}
\end{table}

\begin{table}[H]
\centering
\caption{Capped mean rounds on all matched-seed auctions, charging failures 50 rounds. Entries show mean and paired percentile 95\% bootstrap interval; nine Paths seeds, ten elsewhere.}
\label{tab:architecture-capped-rounds}
\small
\begin{tabular}{lrrr}
\toprule
Domain & FastCombo & Dense & No recurrence \\
\midrule
Arbitrary & 22.43 [21.54, 23.31] & 33.89 [32.43, 35.21] & 25.49 [24.48, 26.48] \\
Paths & 11.92 [11.24, 12.65] & 28.51 [26.99, 29.91] & 14.01 [13.29, 14.76] \\
Regions & 8.90 [8.40, 9.40] & 18.24 [17.23, 19.27] & 10.73 [10.15, 11.32] \\
Scheduling & 6.77 [6.37, 7.19] & 16.25 [15.42, 17.14] & 8.12 [7.73, 8.53] \\
\bottomrule
\end{tabular}
\end{table}

\end{document}